\documentclass[pdflatex,sn-mathphys-num]{sn-jnl}

\usepackage{graphicx}%
\usepackage{multirow}%
\usepackage{amsmath,amssymb,amsfonts}%
\usepackage{amsthm}%
\usepackage{mathrsfs}%
\usepackage[title]{appendix}%
\usepackage{xcolor}%
\usepackage{textcomp}%
\usepackage{manyfoot}%
\usepackage{booktabs}%
\usepackage{algorithm}%
\usepackage{algorithmicx}%
\usepackage{algpseudocode}%
\usepackage{listings}%

\newtheorem{definition}{Definition}
\newtheorem{theorem}[definition]{Theorem}
\newtheorem{lemma}[definition]{Lemma}%
\newtheorem{corollary}[definition]{Corollary}%
\newtheorem{remark}[definition]{Remark}%

\usepackage{cancel}
\usepackage{graphicx,subfigure,wrapfig}
\usepackage{tikz,color}
\usetikzlibrary{shapes,decorations,arrows,calc,arrows.meta,fit,positioning}
\tikzset{
    -Latex,auto,node distance =1.2 cm and 1.2 cm,
    state/.style ={circle, draw, minimum width = 0.3 cm},
    point/.style = {circle, draw,fill,node contents={}},
    bidirected/.style={Latex-Latex},
}
\newcommand{\one}{\mathbf{1}}

\begin{document}

\title[Quantum Advantage in Decision Trees: A Weighted Graph and $L_1$ Norm Approach]{Quantum Advantage in Decision Trees: A Weighted Graph and $L_1$ Norm Approach}


\author*[1,3]{\fnm{Sebastian Alberto} \sur{Grillo}}\email{sgrillo@uaa.edu.py}

\author[1,2]{\fnm{Bernardo Daniel} \sur{Dávalos}}\email{davalos@facen.una.py}
\equalcont{These authors contributed equally to this work.}

\author[3]{\fnm{Rodney Fabian} \sur{Franco Torres}}\email{rfrancot@pol.una.py}
\equalcont{These authors contributed equally to this work.}

\author[4]{\fnm{Franklin} \spfx{de} \sur{Lima Marquezino}}\email{franklin@cos.ufrj.br}
\equalcont{These authors contributed equally to this work.}

\author[2,3]{\fnm{Edgar} \sur{López Pezoa}}\email{epezoa@facen.una.py}
\equalcont{These authors contributed equally to this work.}

\affil*[1]{\orgdiv{Facultad de Ciencia y Tecnología}, \orgname{Universidad Autónoma de Asunción}, \orgaddress{\street{Jejuí 667, entre O’leary y 15 de agosto}, \postcode{1255}, \state{Asunción}, \country{Paraguay}}}

\affil[2]{\orgdiv{Facultad de Ciencias Exactas y Naturales}, \orgname{Universidad Nacional de Asunción}, \orgaddress{\street{Campus de la UNA}, \postcode{1039}, \state{San Lorenzo}, \country{Paraguay}}}

\affil[3]{\orgdiv{Facultad Politécnica}, \orgname{Universidad Nacional de Asunción}, \orgaddress{\street{Campus de la UNA}, \postcode{2111}, \state{San Lorenzo}, \country{Paraguay}}}

\affil[4]{\orgdiv{PESC-COPPE}, \orgname{Universidade Federal do Rio de Janeiro}, \orgaddress{\street{Avenida Horacio Macedo 2030}, \postcode{21941-598}, \state{Rio de Janeiro}, \country{Brazil}}}


\abstract{The analysis of the computational power of single-query quantum algorithms is important because they must extract maximal information from one oracle call, revealing fundamental limits of quantum advantage and enabling optimal, resource-efficient quantum computation.
This paper proposes a formulation of single-query quantum decision trees as weighted graphs.
This formulation has the advantage that it facilitates the analysis of the $L_1$ spectral norm of the algorithm output. This advantage is based on the fact that a high $L_1$ spectral norm of the output of a quantum decision tree is a necessary condition to outperform its classical counterpart. We propose heuristics for maximizing the $L_{1}$ spectral norm, show how to combine weighted graphs to generate sequences with strictly increasing norm, and present functions exhibiting exponential quantum advantage. Finally, we establish a necessary condition linking single-query quantum advantage to the asymptotic growth of measurement projector dimensions.}




\maketitle

\section{Introduction}\label{sec1}

Graph theory plays a fundamental role in several areas of quantum computing, such as quantum error correction and fault tolerance, quantum cryptography and communication, and quantum complexity theory. In the context of quantum algorithm development, graph theory provides a rigorous mathematical framework for formulating and solving problems. In the context of quantum algorithm  design,  graph-theoretic structures provide a rigorous mathematical framework for formulating and analyzing problems in combinatorial optimization, search, and quantum simulation, as we discuss below. 

Quantum walks are a powerful framework for quantum algorithm design, particularly in the context of search problems~\cite{qw-book-renato}. In this model, the vertices of a graph correspond to states within a given search space, while the edges encode the structure of the problem. By leveraging the connectivity properties of the graph, quantum walk-based algorithms can potentially surpass the limitations of Grover’s algorithm~\cite{qw-on-graphs}. Quantum algorithms, particularly hybrid approaches like the Quantum Approximate Optimization Algorithm (QAOA) and the Quantum Walk Optimization Algorithm (QWOA), use these graph-based formulations to navigate solution spaces more efficiently than classical heuristics~\cite{optim-paper-2024}. A second important approach is Hamiltonian simulation, which relies on graph-theoretic methods to provide structured representations of quantum systems~\cite{hamiltonian-simulation-1996}. Many physical models, such as in condensed matter physics and quantum chemistry, can be represented by graphs where vertices correspond to quantum states or subsystems and edges represent interactions dictated by the system's Hamiltonian. Quantum algorithms designed for simulation, such as those based on Trotterization or variational methods, make use of graph representations to improve computational efficiency and accuracy. Graph structures enable the decomposition of complex Hamiltonians into local terms, helping the implementation on quantum hardware and reducing errors.

Quantum search or optimization algorithms can be analyzed in terms of the amount of information they  extract from the input,  abstracting away other  computational costs. The quantum decision tree, or quantum query model, is a useful framework for analyzing quantum algorithms where the complexity is measured by the number of queries to an oracle \cite{barnum2003quantum}. In that sense, a notable class of quantum decision trees is the one that performs only one query on the oracle, since they can be applied to search and optimization problems \cite{grover1997quantum, terhal1998single, grosshans2015factoring}. 
Developing techniques for single-query quantum algorithms constitutes a fundamental challenge in quantum computing, since all relevant information must be extracted from a single interaction with the oracle. This extreme restriction forces algorithm designers to exploit quantum superposition, interference, and phase encoding to their theoretical limits. At the same time, advances in this setting guide the design of optimal and resource-efficient quantum procedures. Insights obtained from this regime often extend to more general quantum algorithms and are especially relevant for near-term quantum devices with limited computational resources \cite{deutsch1992rapid,beals1998quantum}.
In this regard, single-query algorithms are able to solve problems where there is the greatest difference in complexity between classical and quantum decision trees \cite{aaronson2015forrelation}. 

Aaronson et al. \cite{aaronson2015polynomials} proved that a function can be approximated with bounded error by a second-degree multilinear polynomial if and only if it can be approximated with bounded error by the output of a single-query algorithm. This polynomial approximation encompasses the traditional notion of Nisan and Szegedy, and implies that by analyzing the maximum and minimum values of a second-degree multilinear polynomial on a subset of $\left\{ 1,-1 \right\}^n$, we will identify a Boolean function computable by a single-query quantum algorithm with bounded error. In this work, we are particularly interested in identifying functions that can be computed by a single-query quantum algorithm but are very costly for classical query algorithms. It is important to notice that any classical complexity increase when faced with a problem computed by a constant number of queries by quantum means implies an exponential gain.

A high $L_{1}$ spectral norm of the output is a necessary condition for a quantum decision tree to be expensive to simulate  in the classical query model \cite{grillo2019fourier}. Therefore, if a problem is expensive to solve by classical means but is efficiently solved with quantum decision trees, then that problem can only be solved with algorithms whose output has a high $L_{1}$ spectral norm. To this end, this paper proposes a representation of second-degree multilinear polynomials as Weighted Dynamical Graphs (WDG), where the computation of the $L_{1}$ spectral norm  appears directly as the sum of the weights. Therefore, we represent each monomial of the polynomial as an edge and each variable as a vertex of the WDG. However, like any weighted graph, there is also a symmetric matrix representation which we also present as an equivalent associated formulation. The WDG is associated with a function that is equal to the corresponding linear polynomial of second degree except for a constant. The WDG model acts as an alternative model for single-query quantum algorithms through the equivalence between single query quantum decision trees and linear polynomials of second degree, under bounded error \cite{aaronson2015polynomials}. However, an algorithm described according to the WDG model does not imply a direct representation according to the quantum query model.

From the WDG representation, we formulate two optimization problems equivalent to maximizing the $L_{1}$ spectral norm  by computing a partial function under bounded error. The first problem consists directly in maximizing the spectral norm $L_1$ of the WDG, where the difference between the maximum and the minimum of the function associated to the WDG does not exceed 1 as a constraint. However, it is complicated to perform an analysis of the WDG while maintaining the constraint. To address this, we formulate the second problem of minimizing the difference between the maximum and the minimum of the function associated with the WDG, but keeping the spectral norm $L_1$ constant. These optimization problems can be approached with classical optimization methods on matrices. However, we show by means of a lower bound that distributing the weights uniformly among the edges of each vertex is a heuristic that tends to maximize the $L_{1}$ spectral norm. In addition, we show that one or more single-query quantum decision trees can be combined to generate a sequence of single-query quantum decision trees whose $L_{1}$ spectral norm is strictly increasing. Thus, such sequences have the potential to solve problems where quantum algorithms have an exponential advantage over classical algorithms, and we present an example of a function identified in this way where single-query quantum algorithms possess such a superpolynomial advantage.

Beyond analytic properties of the output polynomial, the measurement process has been studied as an independent resource, with results on measurement compression and information extraction showing that reproducing the statistics of a given POVM may require nontrivial classical resources depending on the number of outcomes and the effective dimensionality of the associated projectors~\cite{winter2001compression, massar2000amount}. Moreover, even when the underlying quantum evolution is simple, estimating or simulating measurement statistics can be computationally hard, highlighting that the complexity of measurement alone may act as a bottleneck~\cite{gharibian2019complexity}. Motivated by these perspectives, we analyze the relationship between the computational advantage of single-query algorithms and the number of dimensions of the projectors that define the measurement stage. We identify a necessary condition for a single-query quantum algorithm with a finite number of outputs to have an advantage over classical algorithms. This condition is that there exists an output associated with a projector whose number of dimensions tends to infinity, as does its difference from the total number of dimensions that define the states of the algorithm. This result provides a concrete guide for the design of single-query quantum algorithms.

We summarize the main contributions of the paper in the following items:
\begin{enumerate}
    \item Formulation of single-query quantum decision trees as weighted dynamical graphs. 
    \item Formulation of the $L_{1}$ spectral norm maximization problem for single-query quantum decision trees as weighted graphs.
    \item A technique for combining weighted dynamical graphs, which if applied iteratively yields a sequence of weighted dynamical graphs computing functions with potential advantage over classical algorithms. In addition, an example of a function computed by iterative weighted dynamical graphs is presented, where quantum algorithms have an exponential advantage over classical decision trees.
    \item A necessary condition for a single-query quantum algorithm with a finite number of outputs to have an asymptotic advantage over randomized decision trees.
\end{enumerate}

 While prior work has established connections between polynomial approximations and quantum query complexity \cite{aaronson2015polynomials, grillo2019fourier}, the graph formulation proposed here offers a novel combinatorial perspective. However, it remains unclear how this approach compares to alternative frameworks such as sum-of-squares hierarchies or quantum communication complexity. This paper focuses exclusively on the weighted graph representation and its analytic implications for single-query quantum advantage.

This paper is organized as follows. In Section \ref{sec2} we present previous definitions and theorems from the literature that are applied. In Section \ref{sec3} we present the formulation of WDG for single-query quantum decision trees. In Section \ref{sec4} we formulate the problem of maximizing the $L_{1}$ spectral norm for WDG, how to obtain sequences of WDGs with increasing $L_{1}$ spectral norm and the example of a computed function that has an exponential advantage over classical algorithms. In Section \ref{sec5}, we present the necessary condition for a single-query algorithm with a finite number of outputs to have an asymptotic advantage over classical algorithms.
Finally, in Section \ref{sec6} we present the conclusions on the proposal.

\section{Preliminaries}\label{sec2}


 This section establishes notation and foundational results from quantum query complexity, polynomial approximation, and certificate complexity. We adhere to the standard quantum query model (QQM) \cite{barnum2003quantum} and adopt the Fourier-analytic framework for Boolean functions \cite{o2014analysis}.

The Quantum Query Model (QQM)~\cite{barnum2003quantum} describes algorithms that compute functions whose domain is a subset of $\{0,1\}^{n}$. The states of a QQM algorithm belong to a Hilbert space $\mathcal{H}$ with basis states $\lvert i \rangle \lvert j \rangle$, where $i \in [n]$, $j \in [m]$ for an arbitrary $m$ and we denote $[k]=\left\{ 0,1,..,k \right\}$. There is a query operator defined as
\[
O_{x} \lvert i \rangle \lvert j \rangle = (-1)^{x_{i}} \lvert i \rangle \lvert j \rangle,
\]
where $x \equiv x_0 x_1 \cdots x_n$ is the input, and $x_{0} \equiv 1$ is an ancilla bit. The final state of the algorithm on input $x$ is denoted as
\[
\lvert \Psi^{f}_{x} \rangle 
= U_{t} O_{x} U_{t-1} \cdots O_{x} U_{0} \lvert \Psi \rangle,
\]
where $\{ U_{i} \}$ is a set of unitary operators on $\mathcal{H}$ and $\lvert \Psi \rangle \in \mathcal{H}$. The number of queries or steps is the number of times that $O_{x}$ appears in the final state. 

\begin{definition}[Complete Set of Orthogonal Projectors]
Let $\mathcal{H}$ be a Hilbert space. A collection of projectors $\{P_i\}_{i=1}^n$ on $\mathcal{H}$ is called a \textit{complete set of orthogonal projectors} if it satisfies the following properties:
\begin{enumerate}
    \item $P_i^2 = P_i$ \hfill (idempotence)
    \item $P_i^\dagger = P_i$ \hfill (self-adjointness)
    \item $P_i P_j = 0$ for all $i \neq j$ \hfill (orthogonality)
    \item $\displaystyle \sum_{i=1}^n P_i = I$ \hfill (completeness)
\end{enumerate}
The projectors $\{P_i\}$ form a resolution of the identity and decompose the Hilbert space into mutually orthogonal subspaces.
\end{definition}

The probability of obtaining an output $z\in Z$ is
$\pi_{z}\left(x\right)=\left\Vert P_{z}\left|\Psi_{x}^{f}\right\rangle
\right\Vert ^{2}$, given a CSOP $\{P_i\}$. 
An algorithm computes a function $f:D\rightarrow Z$ within error $\varepsilon$ if
 $\pi_{f\left(x\right)}\left(x\right) \geq 1-\varepsilon$ for all input $\ensuremath{x}\in D\subset\left\{ 0,1\right\} ^{n}$.
 Theorem \ref{teo0} presents an alternative representation of the final state of a QQM algorithm \cite{grillo2018quantum}.
\begin{theorem}\label{teo0}
Given a QQM algorithm, we also introduce the following notation:
\begin{itemize}
\item A product of unitary operators
$\widetilde{U}_{n}=U_{n}U_{n-1}\ldots U_{0}$. 
    \item A CSOP $\left\{
  \bar{P}_{k}:0\leq k\leq n\right\}$, where the range of each $\bar{P}_{i}\left|i\right\rangle
\left|\psi\right\rangle=\left|i\right\rangle
\left|\psi\right\rangle$ and $\bar{P}_{i}\left|j\right\rangle
\left|\psi\right\rangle=0$ for $ i\neq j$ and any state $\left|\psi\right\rangle.$
    \item Given a fixed $j$, the CSOP $\left\{ \widetilde{P}_{j,k}:0\leq
  k\leq n\right\}$, where ${\widetilde{P}_{j,i}=\widetilde{U}_{j}^{\dagger}\bar{P}_{i}\widetilde{U}_{j}}$.

\end{itemize}

If the QQM algorithm applies $t$ queries over $x=x_0x_1...x_n$, then:

    \begin{equation}
  \label{desc}
  \widetilde{U}_{t-1}^{\dagger}O_{x}U_{t-1} \ldots U_{1}O_{x}U_{0}\left|\Psi\right\rangle =\sum_{k_{t-1}=0}^{n} \ldots \sum_{k_{0}=0}^{n} (-1)^{\sum_{i=0}^{t-1}x_{k_{i}}}\widetilde{P}_{t-1,k_{t-1}}\widetilde{P}_{t-1,k_{t-1}}...\widetilde{P}_{1,k_1}\widetilde{P}_{0,k_{0}}\left|\Psi\right\rangle.
\end{equation}
\end{theorem}

For each $b\in \left \{ 0,1 \right \}^{n}$, we define $\chi_{b}:\left\{ -1,1\right\} ^{n}\rightarrow\left\{ -1,1\right\}$ such that $\chi_{b}\left(x\right)=\underset{i}{\prod}b_{i}x_{i}$. This family of functions is an orthonormal basis for the space of functions of the form $f:\left\{ 1,-1\right\} ^{n}\rightarrow\mathbb{R}$, see \cite{o2014analysis}. Therefore, any function $f: \{0,1\}^n \to \mathbb{R}$ has a representation as a linear combination $$f = \sum_{b \in \{0,1\}^n} \alpha_b \chi_b,$$ and we denote the Fourier $1$-norm of $f$ as $$L(f) = \sum_{b \in \{0,1\}^n} |\alpha_b|.$$

In the same way that the output probability of an algorithm approximates a given function, Definition \ref{polaprox} extends this notion to polynomials.
\begin{definition}\label{polaprox}
 We say that a polynomial $p$ approximates a function $f$ with error bounded by $\varepsilon$, if $\left | p\left ( x \right ) -f\left ( x \right )\right |< \varepsilon $ for all $x$ in the domain of $f$.
\end{definition}
\begin{theorem}
\label{Th1}
There is a $\varepsilon<1/2$ and a ${\varepsilon}'<1/2$, such that a partial boolean function $f$ can be approximated by a degree-2 polynomial with error bounded by ${\varepsilon}<1/2$, if and only if,  $f$ is computable by a 1-query quantum algorithm with error bounded by ${\varepsilon}'<1/2$. \cite{aaronson2015polynomials}
\end{theorem}

Theorem \ref{Th1} implies that we can discover new problems that are decidable by 1-query quantum algorithms, by analyzing degree-2 polynomials.

A certificate is a string of bits that takes the shortest length of bits in the available string on the variable such that it verifies an algorithm. Formally is defined as follows \cite{aaronson2008quantum}:

\begin{definition}[Certificate complexity]\label{def:CC}
    Let $f:\{-1,1\}^n\to \{0,1\}$ be a partial function and $x\in \{-1,1\}^n$. Let us define a 0-certificate, denoted by $C^S_0(f)$, for $x$ as the subset $S\subseteq \{-1,1\}^n$ such that $f(x')=0$ for every $x'\in S$, where $x'=x_{|S}$. 
Similarly, a 1-certificate, denoted by $C^S_1(f)$, is the subset $S\subseteq \{-1,1\}^n$ such that $f(x')=1$ for every $x'\in S$.
    The \emph{certificate complexity} for $f$ is defined as 
    \begin{equation}\label{eq:CC}
        C(f) := \min_{k\in [1]} \max_{S \subseteq \{-1,1\}^n} C_k^S(f)
    \end{equation}
\end{definition}

In a randomized decision tree, the input location to be queried is chosen probabilistically \cite{arora2007}.

It is asserted that a randomized decision tree is said to compute $f$ with bounded error if its output is equal to $f(x)$ with probability $1$ and rejects with probability at least $1-\varepsilon$, if $f(x)\neq f(y)$ for all $y \in \{-1, 1\}^n$ \cite{buhrman2002complexity}. 
We denoted by $R_{\varepsilon}\left(f\right)$ for the \emph{randomized complexity with error at most less than $\epsilon$} \cite{buhrman2002complexity}, if it is the minimal number of queries that a classical algorithm applies in order to compute $f$ within error $\varepsilon$ by a randomized decision tree \cite{grillo2019fourier}.

Aaronson's results \cite{aaronson2008quantum} provide a compelling basis for proving the validity of the following equation
\begin{equation}\label{eq:relation_R_C}
    R_{\varepsilon}(f) = \Omega\left(\sqrt{C(f)}\right).
\end{equation}

\begin{theorem}
\label{TH2}
 Consider $D\subset \left\{ 0,1\right\}^{n}$ and a function $f:D\rightarrow\left\{ 0,1\right\}$
 that is $\varepsilon$-approximated by a polynomial $p:\mathbb{R}^n\rightarrow \mathbb{R}$. If $\deg\left(p\right)\leq 2t$, then
 \begin{equation}\label{1eq8}
\frac{R_{\varepsilon}\left(f\right)}{2t}\in O\left(L\left(p\right)^{2}\right),
\end{equation}
as $n\to \infty$.
\end{theorem}
Finally, Theorems \ref{Th1} and \ref{TH2} implies that we can find problems where 1-query quantum algorithms have an exponential advantage over classical algorithms by analyzing the partial boolean functions that can be approximated by degree-2 polynomials whose $L$ norm grows asymptotically on $n$.

\section{Polynomials of degree at most two as weighted graphs}\label{sec3}

In this section we formulate quantum algorithms as weighted graphs and also present their matrix representation. From here we maintain the ancilla bit $x_0=1$. The following theorems show how to analyze degree-2 polynomials as graphs.
\begin{theorem}
 If $p:\left\{ 1,-1\right\} ^{n + 1}\rightarrow\left[0,1\right]$ is a multilinear polynomial of degree at most two, then $p$ has a unique representation \begin{equation}\label{sum}
p=\underset{b\in \left \{ 0,1 \right \}^{n}}{\sum}\alpha_{b}\chi _{b},
\end{equation} such that $\left|b\right|\leq2$ and $\alpha_{b} \in \mathbb{R}$.
\end{theorem}

\begin{proof}
Follows from the definition of $\chi_{b}$, for more details see \cite{o2014analysis}.
\end{proof}

For example, take  $p:\left\{ 1,-1\right\} ^{6}\rightarrow\left[0,1\right]$ denoted as 

\begin{equation}\label{ej1}
    p\left( x \right) = \frac{1}{8}\left( x_0x_2+x_0x_5-x_{1}x_{2}-x_{4}x_{3}\right)+\frac{1}{2}.
\end{equation}

We can see that each monomial is a function of the form $\chi_{b}\left(x\right)=\underset{i}{\prod}b_{i}x_{i}$ for some $b\in \left \{ 0,1 \right \}^{n+1}$. Then, we can simplify every polynomial such that two monomials are not assigned the same function $\chi_{b}\left(x\right)$, which gives us a unique representation for each polynomial of the form given by Equation \ref{sum}.

The following definition is an alternative representation of degree-2 polynomials by weighted graphs, as Theorem \ref{t1} shows.
\begin{definition}
 Let $G=\left(V,E\right)$ be a graph with no multiple edges, where $V$ and $E$ are the sets of vertices and edges, respectively. Denote $w:E\rightarrow \mathbb{R}$ and $x\in\left \{ -1,1 \right \}^{n+1}$  (considering  the ancilla bit $x_{0}$  equal to 1), we say that $D=\left(G,w\right)$ is a weighted dynamical graph (WDG) and define \begin{equation}\label{D2}
     g_{D}\left(x\right)=\sum_{e\in E}s\left ( e,x \right )w\left ( e \right )
 \end{equation} as the value of $D$ over $x$, where if $e=\left\{ v_{i},v_{j} \right\}$ then $s\left ( e,x \right )=x_{i}x_{j}$ and if $e=\left\{ v_{i} \right\}$ then $s\left ( e,x \right )=x_{i}x_{0}$.
\end{definition}

This graph representation allows us to visualize the algebraic structure of the polynomial and facilitates the computation of its $L_{1}$ norm as the sum of the absolute values of the edge weights, which is central to our quantum advantage analysis.

We highlight that the graph has $n+1$ vertices, where the vertex of the ancilla bit remains constant.

\begin{theorem} 
\label{t1}
 There is a unique multi-linear polynomial  $p:\left\{ 1,-1\right\} ^{n + 1}\rightarrow\left[0,1\right]$ of degree at most two if and only if there is a WDG $D$ such that $p\left(x\right)=g_{D}\left(x\right)+K$ for some constant $K$.
\end{theorem}
\begin{proof} Notice that
\begin{enumerate}
    \item For $|b|=2$ and $e=\left\{ i,j \right\}$,  if $b_{i}=b_{j}=1$ then $s\left ( e,x \right )=\chi _{b}\left ( x \right )$ and $\alpha _{b}=w\left ( e \right )$;
    \item For $|b|=1$ and $e=\left\{ i \right\}$, if $b_{i}=b_{j}=1$ then $s\left ( e,x \right )=\chi _{b}\left ( x \right )$ and $\alpha _{b}=w\left ( e \right )$; and
    \item For $|b|=0$ we have $K=\chi _{b}\left ( x \right )$.
\end{enumerate}

To prove the uniqueness, suppose that there are two polynomials $p_1, p_2:\{-1,1\}^{n+1} \rightarrow [0,1]$ and constants $K_1,K_2$ such that $p_1(x) = g_D(x) + K_1$ and $p_2(x) = g_D(x) + K_2$ are equals at $x'\in \{-1,1\}^{n+1}$,
i.e. $p_1(x')=p_2(x')$, then $g_D(x') + K_1 = g_D(x') + K_2 $ implies $K_1=K_2$. 
\end{proof}

 The equivalence established in Theorem \ref{t1} enables the translation of analytic pro\-perties of polynomials, such as the $L_{1}$ norm, into combinatorial properties of weighted graphs. This connection is central to the optimization problems formulated in Section \ref{sec4}.

For example, for the polynomial of Equation \ref{ej1} we take $G=\left(V,E\right)$ such
$$V=\left \{ v_0, v_{1}, v_{2}, v_{3}, v_{4}, v_{5} \right \},$$ 
$$E=\left \{\{v_0,v_2\}, \{v_0,v_5\}, \left \{ v_{1}, v_{2} \right \}, \left \{ v_{3}, v_{4} \right \} \right \}$$ and $w:E\rightarrow \mathbb{R}$ is denoted in Table \ref{T1}. Then taking the Equation \ref{D2}, we have that 
\begin{align*}
g\left(x\right)=\frac{1}{8}\left(x_0x_2+x_0x_5-x_{1}x_{2}-x_{4}x_{3}\right).    
\end{align*}
Theorem \ref{t1} gives an alternative representation of 2-grade polynomials as weighted graphs. Maintaining our original example, the polynomial from Equation \ref{ej1} has a weighted graph representation that we show in Figure \ref{G1}.

\begin{table}[h]
\centering
\begin{tabular}{c|c}
Edge & Value of $w$ \\ \hline
  $\{v_0,v_2\}$ & $\frac{1}{8}$     \\
  $\{v_0,v_5\}$ & $\frac{1}{8}$     \\
   $\left \{ v_{1}, v_{2} \right \} $  &     $-\frac{1}{8}$     \\
    $\left \{ v_{3}, v_{4} \right \}$  &     $-\frac{1}{8}$    
\end{tabular}
\caption{Weights for each edge from the WDG of Equation \ref{ej1}.\label{T1}}
\end{table}

\begin{figure}[h]

\centering
\begin{tikzpicture}
    \node[state] (v5) at (-1,-0.2) {$v_5$};
    \node[state] (v1) at (1.5,1) {$v_1$};
    \node[state] (v2) at (2.5,0) {$v_2$};
    \node[state] (v4) at (-0.5,-1.3) {$v_4$};
    \node[state] (v3) at (1.8,-0.8) {$v_3$};
    \node[state] (v0) at (-0.75,1.3) {$v_0$};

    \path[bidirected,blue] (v1) edge node{$-1/8$} (v2);
    \path[bidirected,blue] (v4) edge node{$-1/8$} (v3);
    \path[bidirected, blue] (v2) edge node{$1/8$} (v0);
    \path[bidirected, blue] (v5) edge node{$1/8$} (v0);
  
\end{tikzpicture}
\caption{Weighted graph for Equation \ref{ej1}. \label{G1}}
\end{figure}
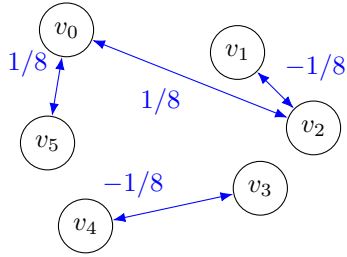

An important property of a WGD $D$ is the difference between the maximum and minimum of $g_D$, which we formalise below.

\begin{definition}\label{def:difmax}
Let $D=(G,w)$ be a WDG. We define the maximum difference of $D$, denoted by $\delta\left(D\right)$, as the value 
\begin{equation}
    \delta(D) := \max_{x\in\{-1,1\}^{n+1}}g_D(x) - \min_{x\in\{-1,1\}^{n+1}} g_D(x).
\end{equation}
\end{definition}
Taking the example in Figure \ref{G1}, we can see that $g_D$ reaches a maximum value of $\frac{1}{2}$ for $(1, -1, 1, 1, -1, 1)$ and a minimum value of $-\frac{1}{2}$ for $(1, -1, -1, 1, 1, -1)$, thus $\delta(D)=1$.

Finally, we note that a WGD $D$ has a matrix representation that allows us to calculate its corresponding function $g_D$ in a straightforward way.

\begin{definition}\label{mat}
    Let $D=(G,w)$ be a WDG with $n$ vertices. We define the matrix associated to $D$ as a  $(n+1)\times (n+1)$ real matrix $M^D = \left(M_{ij}^{D}\right)$, with $0\leq i, j \leq n$, such that:
    \begin{itemize}
        \item If $ij = 0$ and $k = \max\left\{i, j\right\} > 0$, then $M^D_{ij}=w\left( \left\{ v_k \right\} \right)$.
        \item If $i=j$ then $M^D_{ij}=0$.
        \item If $i,j>0$ and $i \neq j$, then $M^D_{ij}=M^D_{ji}=w\left( \left\{ v_i, v_j \right\} \right)$.
    \end{itemize}
    
\end{definition}

From Definition \ref{mat} we can observe that for $x \in \left\{ 1,-1 \right\}^{n + 1}$ we have $g_{D}\left( x \right) =\frac{1}{2} xM^{D}x^{t}$. For example, the WGD from Figure \ref{G1} implies that the matrix associated to $D$ is:

$$ M^{D}=\begin{bmatrix}
0 & 0 & 1/8 & 0 & 0 & 1/8 \\
0 & 0 & -1/8 & 0 & 0 & 0 \\
1/8 & -1/8 & 0 & 0 & 0 & 0 \\
0 & 0 & 0 & 0 & -1/8 & 0 \\
0 & 0 & 0 & -1/8 & 0 & 0 \\
1/8 & 0 & 0 & 0 & 0 & 0 \\
\end{bmatrix}.$$

In particular if $\mathbf{1}=(x_0, x_1,\dots,x_n)$ is such that $x_i=1$ for each $i=0, 1, \dots, n$, the total weight of $D$ is determined by $w(M^D)=g_D(\one)$.

\section{The $L_{1}$ spectral norm maximization problem.}\label{sec4}

From Theorems \ref{Th1} and \ref{TH2}, we can discover new problems where quantum algorithms have a potential advantage over classical algorithms. This would be done by searching for linear polynomials of the second degree or WGDs bounded between 0 and 1, which have a high $L_1$ norm. This search can be restricted by defining a partial function that should be computed. Thus the search is formulated as identifying which partial functions containing the first one can be computed if we maximize the $L_1$ norm. We define this  optimization problem below.
\begin{definition} \label{max}
    Let $f$ be a function such that $f:S\rightarrow \left [ 0,1 \right ]$ and $S\subset \left\{1,-1 \right\}^{n}$. Considering a WDG $D=\left ( G, w \right )$ with $G=\left ( V, E \right )$ and given some $\varepsilon>0$: \begin{itemize}
        \item Find $C$ and $w$ that maximize \begin{equation}
            \sum_{e\in E} \left|w\left ( e \right ) \right|.
        \end{equation}
        \item Subject to:
        \begin{enumerate}
            \item[i)] constraint $\left|g_{D}\left ( x \right )-f\left ( x \right )+C \right|< \varepsilon $ for each $x \in S$  and
            \item[ii)] constraint $\delta(D)=1.$
        \end{enumerate}
    \end{itemize}
\end{definition}

The optimization problem from Definition \ref{max} can be formulated as an equivalent minimization problem where we replace constraint ii), which simplifies the problem and is presented below.

\begin{definition}\label{min}
    Let $f$ be a function such that $f:S\rightarrow \left [ 0,1 \right ]$ and $S\subset \left\{1,-1 \right\}^{n}$. Considering a WDG $D=\left ( G, w \right )$ with $G=\left ( V, E \right )$ and given some $\varepsilon>0$: \begin{itemize}
        \item Find $C$ and $w$ that minimise $\delta(D).$
        \item Subject to:
        \begin{enumerate}
            \item[i)] constraint $\left|g_{D}\left ( x \right )/\delta(D)-f\left ( x \right )+C \right|< \varepsilon$ for each $x \in S$  and
            \item[ii)] constraint \begin{equation}
            \sum_{e\in E} \left|w\left ( e \right ) \right|=1.
        \end{equation}
        \end{enumerate}
    \end{itemize}
\end{definition}

These optimization problems are formally equivalent, but from an algorithmic standpoint they can be approached with different techniques: the first as a maximization problem with bounded difference constraint, the second as a minimization problem with fixed norm. In practice, the minimization formulation may be more numerically stable. The equivalence of the  optimization problems of Definitions \ref{max} and \ref{min} starts from the following observation. If we solve the minimization problem of $\delta(D)$ by finding $w$, then we solve the maximization problem by defining $w'=w/\left(\min  \delta(D) \right )$ as the weight function for the new WGD. We can also transform a maximizing solution to a minimizing solution by multiplying the weights $w$ by a factor. However, we are only interested in transforming the  minimization problem to the  maximization problem. Depending on which case is simpler to  analyze, we will present results for the maximization or  minimization problem based on the equivalence mentioned above.

It may be noted that any solution satisfying the constraints of Definitions \ref{max} or \ref{min} for $\varepsilon=0$, implies a single-query algorithm that computes the function $f$ under a bounded error $\varepsilon'$ according to Theorem \ref{Th1}. 

The problem of minimization of $\delta(D)$ can also be studied from the necessary properties that the weight function $w$ must fulfill to be bounded. As such, these properties serve as a guide to design single-query quantum algorithms that have potential advantage over classical ones.

We introduce the following notation before the next theorem. Given a graph $G = \left(V, E\right)$, we denote $E\left(v\right)$ as the set of edges underlying $v \in V$.

\begin{theorem}[Lower bound]\label{cot}
Let $G = \left(V, E\right)$ be a graph and $D = \left(G, w\right)$, if
\begin{equation}\label{epsilon}
    \max_{v\in V}\left\{ \sum_{e\in E\left ( v \right )}\left|w_{e}\right|\right\}=\epsilon
\end{equation}
then $\delta(D)\geq 2\epsilon$.
\end{theorem}

\begin{proof}
    Let $v_0 \in V$ be such that $\sum_{e \in E(v_0)} |w_e| = \epsilon$ is the value satisfying Equation \ref{epsilon}, and $x^+, x^- \in\{-1,1\}^n$ are such that $|w_e| = w_e s(e,x^+) $ and $-|w_e| = w_e s(e,x^-)$ for $e\in E(v_0)$. Since $E(v) \subset E$, then $\sum_{e \in E(v)} |w_e| \le \sum_{e \in E} |w_e|$. Then,
    \begin{align*}
        \delta (D) & = \max_{x \in \{-1,1\}^n} \sum_{e\in E} w_e s(e,x) - \min_{x \in \{-1,1\}^n} \sum_{e \in E} w_e s(e,x)\\
        & \geq \sum_{e\in E} w_e s(e,x^+) - w_e s(e,x^-)\\
        & = 2\sum_{e \in E(v_0)} |w_e| + \underbrace{\sum_{e\in E \setminus E(v_0) } w_e s(e,x^+) -  w_e s(e,x^-)}_{\geq 0}\\
        & \geq 2\sum_{e\in E(v_0)} |w_e|\\
        & = 2\epsilon
    \end{align*}
\end{proof}

Theorem \ref{cot} implies that if the weights assigned to the edges tend to be constant, then $\delta(D)$ will tend to be smaller. Note that this theorem is stated for the minimization problem. In the case of the maximization problem, it equivalently implies that if the weights assigned to the edges tend to be constant, then the $L_1$ norm will tend to be larger.

Next, we will show how to combine two WDG to obtain another WDG with a higher $L_1$ norm, but first we introduce properties of the $L_1$ norm by using Lemma \ref{lem:normL1}. For two matrices with the same order $A$ and $B$,  the Hadamard product $A\odot B$ is defined by $(A\odot B)_{ij}=(A_{ij}\cdot B_{ij})$.
For a matrix $A$ and $a \in \mathbb{R}$ we denote the Hadamard power where $E=A^{\circ a}$ implies $E_{ij}:=A_{ij}^{a}$.

\begin{lemma}[$L_1$ norm properties]\label{lem:normL1}
Let $D=(G,w)$ a WDG and $M^D$ his matrix associated and $K$ a constant. Then
\begin{itemize}
    \item[(i)] $L(g_D)=w((M^D\odot M^D)^{\circ 1/2})$
    \item[(ii)] $L(g_D + K)=L(g_D) + |K|$
    \item[(iii)] $L(g_{D'\otimes D''}) = 2L(g_{D'})L(g_{D''})$ if $M^D=M'\otimes M''$
    \item[(iv)] $L(g_D\cdot g_{D'})=L(G_D)L(g_{D'})$ 
    \item[(v)] $L(K\cdot g_D)=|K|\cdot L(g_D)$
\end{itemize}
where $K$ is a constant.
\end{lemma}

\begin{proof}
(i) Let $M_{abs}=(M^D\odot M^D)^{\circ 1/2}$ and $M^D=(w(v_{ij}))$ be
like Definition \ref{mat} then $(M_{abs})_{ij}=(w(v_{ij})\cdot w(v_{ij}))^{1/2}=|w(v_{ij})|$. Therefore, $w(M_{abs})=\sum_{i<j}|w(v_{ij})|=g_{D'}(\one)$ is the $L_1$ norm of $g_D$.

(ii) Since (i) we can rewrite  $g_{D'}(x)=g_D(x)+K$, then
\begin{align*}
L(g_D+K) &=\sum_{i<j}|w(v_{ij})|+|K|=L(g_D)+|K|
\end{align*}

(iii) By construction and definitions,
\begin{align*}
    L(g_{D'}\otimes g_{D''}) 
    & = w\left(((M^{D'}\otimes M^{D''})\odot (M^{D'}\otimes M^{D''}))^{\circ (1/2)}\right) \\
    & =w\left(((M^{D'}\odot M^{D'})\otimes (M^{D''}\odot M^{D''}))^{\circ (1/2)}\right)\\ 
    &= w\left((M^{D'}\odot M^{D'})^{\circ (1/2)}\otimes (M^{D''}\odot M^{D''})^{\circ (1/2)}\right) \\
    &= \cfrac{1}{2}(\one'\otimes \one'') \left((M^{D'}\odot M^{D'})^{\circ (1/2)}\otimes (M^{D''}\odot M^{D''})^{\circ (1/2)}\right) (\one'\otimes \one'')^t\\
    &=2\left(\cfrac{1}{2}\one'(M^{D'}\odot M^{D'})^{\circ (1/2)}\one'^t\right)\left(\cfrac{1}{2}\one''(M^{D''}\odot M^{D''})^{\circ (1/2)}\one''^t\right)\\
    &= 2w\left((M^{D'}\odot M^{D'})^{\circ (1/2)}\right)w\left((M^{D''}\odot M^{D''})^{\circ (1/2)}\right)\\
    &=2L(g_{D'})L(g_{D''})
\end{align*}
where the properties of Hadamard product and tensor product are available in \cite{STYAN1973217}.

(iv) For $G=(V,E)$ and $G'=(V',E')$ let's consider that $E$ and $E'$ are disjoint. If $M^D$ and $M^{D'}$ are the matrix representations of $D$ and $D'$ respec., then
\begin{align*}
    L(g_D \cdot g_{D'})
    & = \sum_{e\in E \wedge e'\in E'} |w(e)\cdot w(e')| \\
    & = \sum_{e\in E}\sum_{e\in E'} |w(e)|\cdot |w(e')| \\
    & = \left(\sum_{e\in E} |w(e)|\right)\left(\sum_{e\in E'} |w(e')|\right)\\
    &= L(g_D)\cdot L(g_{D'})
\end{align*}

(v) Particular case of (iv), if $g_D(x)=K$ is the constant function then we can have $M^D=(2K/n(n-1))(\one \cdot\one^t-\mathbb{I}_n)$ is his $n\times n$ matrix representation. So $L(g_D)=\cfrac{2|K|}{n(n-1)}w(M^D)=\cfrac{2|K|}{n(n-1)}\cfrac{n(n-1)}{2}=|K|$.
\end{proof}

Lemma \ref{lem:normL1} provides the algebraic tools to manipulate $L_{1}$ norms under graph operations. In particular, property (iii) shows that the Kronecker product of WDG matrices multiplicatively increases the $L_{1}$ norm, which is exploited in Theorem \ref{rec} to generate sequences with exponentially growing norms.
From here, we consider binary functions where a first class corresponds to the inputs where a WDG reaches its maximum and a second class corresponds to the inputs where the WDG reaches its minimum value. Let there be two WDG and their corresponding binary functions. We can construct a third WDG whose function has a domain $D$, which consists of the Kronecker products of the domains of the previous binary functions. For this new function, the first class consists of the Kronecker products of the inputs that reach the first classes for the previous functions, and the second class consists of the rest of $D$. This new Boolean function can be seen as an AND defined on the tensor Kronecker. Theorem \ref{rec} shows how to construct this WDG and what conditions it must satisfy for the $L_1$ norm to grow in relation to the previous WDG. Consider the following terms:
\begin{itemize}
    \item Let $D$ and $ D'$ be WGDs such that $f\left ( x \right )=\left (\cfrac{1}{2}xM^{D}x^{t}+K\right )$ and $f'\left ( x \right )=\left ( \cfrac{1}{2}xM^{D'}x^{t}+K'\right )$ for $K, K'\in \mathbb{R}$, $S\subset \left\{ 1,-1\right\}^{n}$ and $T\subset \left\{ 1,-1\right\}^{m}$, where $f:S\rightarrow \left [ 0,1 \right ]$ and $f':T\rightarrow \left [ 0,1 \right ]$ respectively.
    \item Denote $S^{+}=\left\{ x\in S:f\left ( x \right )=1\right\}$, $S^{-}=\left\{ x\in S:f\left ( x \right )=0\right\}$, $T^{+}=\left\{ x\in T:f'\left ( x \right )=1\right\}$ and $T^{-}=\left\{ x\in T:f'\left ( x \right )=0\right\}$. 
\end{itemize}
 \begin{theorem}\label{rec}
There is a WGD $D''$ such that 
\begin{itemize}
    \item[(i)] $L\left ( g_{D''} \right )=\left( L\left ( g_{D} \right )+ \left | K\right |\right)\left( L\left ( g_{D'} \right)+ \left | K'\right |\right)-\left |KK'\right |$; and 
    \item[(ii)] 
    $g_{D''}\left ( x \right )=f''\left ( x \right )+K''$ for $x\in ((S^{+} \otimes T^{+}) \cup (S^{-} \otimes T^{-})\cup (S^{+} \otimes T^{-})\cup (S^{-} \otimes T^{+}))$ and some $K''\in \mathbb{R}$, where the function $f''$ satisfies a) $f''(\omega)=1$ for $\omega\in (S^+\otimes T^+)$ and b)
    $f''(\omega)=0$ for $\omega\in ((S^{+} \otimes T^{-}) \cup (S^{-} \otimes T^{+}) \cup (S^{-} \otimes T^{-}))$.
\end{itemize}
\end{theorem}
\begin{proof}
    Denote $M^{D''}= \tfrac{1}{2}\left (\left (M^{D} + 2K\cdot \cfrac{I_{n}}{n}\right)\otimes \left (M^{D'} + 2K'\cdot \cfrac{I_{m}}{m}\right)\right)-\left (2 KK'\cdot I_{n}\otimes \cfrac{I_{m}}{mn}\right)$ as the matrix associated to some WGD $D''$, where $\alpha=$. 
    Notice
    that, taking $x''=x\otimes x'\in \{1,-1\}^{nm}$, 
    \begin{align*}
  g_{D''}(x'') &= \tfrac{1}{2}(x\otimes x')M^{D''}(x\otimes x')^{t} \\
  &=\left(\tfrac{1}{2}xM^Dx^t+\tfrac{K}{n}xI_nx^t\right)
  \left(\tfrac{1}{2}x'M^{D'}x'^t+\tfrac{K'}{m}x'I_mx'^t\right) -
  \tfrac{KK'}{mn}(xI_nx^t)(x' I_m x'^t) \\
  &=(g_D+\tfrac{K}{n}\|x_D\|^2)(g_{D'}+\tfrac{K'}{m}\|x_{D'}\|^2) - 
  \tfrac{KK'}{mn}(\|x_D\|^2)
  (\|x_{D'}\|^2)\\
  &=(g_D+\tfrac{K}{n}n)(g_D+\tfrac{K'}{m}m)-\tfrac{KK'}{mn}nm\\
  &= f(x)f^{\prime}(x) - KK^{\prime} \\
\end{align*}
Then,
\begin{equation}\label{eq:f.f'}
    f''(x\otimes x')=f(x)f'(x')
\end{equation}
    (i) By Lemma \ref{lem:normL1} and for the above equation, we have
    \begin{align*}
    L(g_{D''})+|KK'| &=L(g_{D''}+KK') \\
            &=L(f\cdot f') \\ 
            &= L(f)L(f')    \\
            &=L(g_D+K)L(g_{D'}+K')\\
            &=(L(g_D) + |K|)(L(g_{D'}) + |K'|)
    \end{align*}

    This proves the statement.
    
    (ii) Since Equation \ref{eq:f.f'}, 
    \begin{itemize}
        \item[a)] If $x''\in\left (S^{+}\otimes T^{+}\right)$, then $f(x)=1$ and $f'(x')=1$, so $f''(x\otimes x')=1$ implies $g_{D''}(x'')=1-K''$.
        \item[b)] If $x''\in((S^{+} \otimes T^{-}) \cup (S^{-} \otimes T^{+}) \cup (S^{-} \otimes T^{-}))$ we have $f(x)=0$ or $f'(x')=0$, then $f''\left ( x\otimes x' \right )=0.$ 
    \end{itemize}
    Therefore, we prove the preposition by taking $K''=-KK'$.    
\end{proof}

\begin{remark}
    A WGD $D=(G,w)$ can have any width, however we want $g_D\in [-K,1-K]$, so that $f:S\to [0,1]$. 
    In order to guarantee that this always happens, we redefine $D_{\delta}=(G,w_{\delta})$ such that his matrix was specified as $M^{D_\delta}=\tfrac{1}{\delta(D)}M^D$.
\end{remark}

Let's consider an example to illustrate the Theorem \ref{rec}.
Let $D$ and $D'$ be two WGD depicted in Figure \ref{fig:DD'}.

\begin{figure}[h]
\centering
\subfigure[$D=(G,w)$]{
\begin{tikzpicture}
    \node[state] (v0) at (0,1.5) {$v_0$};
    \node[state] (v1) at (-1.5,0) {$v_1$};
    \node[state] (v2) at (1.5,0) {$v_2$};

    \path[bidirected,blue] (v1) edge node{$\tfrac{1}{3}$} (v0);
    \path[bidirected,blue] (v2) edge node{$-\tfrac{1}{6}$} (v0);
  
\end{tikzpicture}
}
\subfigure[$D'=(G',w')$]{
\begin{tikzpicture}
    \node[state] (v0) at (0,1.5) {$v_0'$};
    \node[state] (v1) at (-1.5,0) {$v_1'$};
    \node[state] (v2) at (1.5,0) {$v_2'$};

    \path[bidirected,blue] (v1) edge node{$\tfrac{1}{4}$} (v0);
    \path[bidirected,blue] (v0) edge node{$\tfrac{1}{6}$} (v2);
    \path[bidirected, blue] (v1) edge node{$-\tfrac{1}{4}$} (v2);
  
\end{tikzpicture}}
\caption{Two weighted graph.}
\label{fig:DD'}
\end{figure}
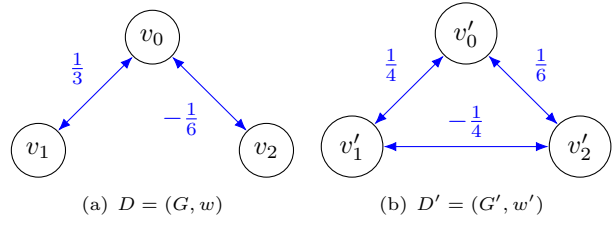

Their matrix representations are
$$M^D=\begin{bmatrix}
    0 & 1/3 & -1/6 \\
    1/3 & 0 & 0 \\
    -1/6 & 0 & 0
\end{bmatrix} \qquad , \qquad 
M^{D'}=\begin{bmatrix}
    0 & 1/4 & 1/6 \\
    1/4 & 0 & -1/4 \\
    1/6 & -1/4 & 0
\end{bmatrix}
$$
and their $L_{1}$ norms are $L(g_D)=\tfrac{1}{2}$ and $L(g_{D'})=\tfrac{2}{3}$. Then $f(x)=\tfrac{1}{3}x_1x_0-\tfrac{1}{6}x_2x_0+\tfrac{1}{2}$ and
$f'(x)=\tfrac{1}{4}x_1x_0+\tfrac{1}{6}x_2x_0-\tfrac{1}{4}x_1x_2+\tfrac{2}{3}$.

Moreover $(1,1,-1)\in S^+$ and $(1,1,1)\in T^+$.

By construction of the proof in Theorem \ref{rec}, let define $D''$ with
\begin{align*}
M^{D''} &= \frac{1}{2}\left(
\begin{bmatrix}
    0 & 1/3 & -1/6 \\
    1/3 & 0 & 0 \\
    -1/6 & 0 & 0
\end{bmatrix}+2\cdot\cfrac{1}{2}
\begin{bmatrix}
    1 & 0 & 0\\
    0 & 1 & 0\\ 
    0 & 0 & 1
\end{bmatrix}/3
\right)\otimes \left(
\begin{bmatrix}
    0 & 1/4 & 1/6 \\
    1/4 & 0 & -1/4 \\
    1/6 & -1/4 & 0
\end{bmatrix}+2\cdot\cfrac{2}{3}
\begin{bmatrix}
    1 & 0 & 0\\
    0 & 1 & 0\\ 
    0 & 0 & 1
\end{bmatrix}/3\right)\\
&-2\cdot\frac{1}{2}\cdot\frac{2}{3} \left(\begin{bmatrix}
    1 & 0 & 0 \\
    0 & 1 & 0 \\
    0 & 0 & 1
\end{bmatrix}\otimes \begin{bmatrix}
    1 & 0 & 0 \\
    0 & 1 & 0 \\
    0 & 0 & 1
\end{bmatrix}\right)/9\\
&=\left(
\begin{bmatrix}
    1/6 & 1/3 & -1/6 \\
    1/3 & 1/6 & 0 \\
    -1/6 & 0 & 1/6
\end{bmatrix}
\otimes
\begin{bmatrix}
    4/9 & 1/4 & 1/6 \\
    1/4 & 4/9 & -1/4 \\
    1/6 & -1/4 & 4/9
\end{bmatrix}\right) - \frac{2}{27}\mathbb{I}_9\\
&=\begin{bmatrix}
    0 & \tfrac{1}{24} & \tfrac{1}{36} & \tfrac{1}{27} & \tfrac{1}{24} & \tfrac{1}{36} & - \tfrac{1}{27} & -\tfrac{1}{24} & -\tfrac{1}{72}\\
    \tfrac{1}{24} & 0 & -\tfrac{1}{24} & \tfrac{1}{24} & \tfrac{1}{27} & -\tfrac{1}{24} & -\tfrac{1}{48} & -\tfrac{1}{27} & \tfrac{1}{48}\\
    \tfrac{1}{36} & -\tfrac{1}{24} & 0 & \tfrac{1}{36} & -\tfrac{1}{24} & \tfrac{2}{27} & -\tfrac{1}{72} & \tfrac{1}{48} & -\tfrac{1}{27}\\
    \tfrac{2}{27} & \tfrac{1}{24} & \tfrac{1}{36} & 0 & \tfrac{1}{24} & \tfrac{1}{36} & 0 & 0 & 0 \\
    \tfrac{1}{24} & \tfrac{2}{27} & -\tfrac{1}{24} & \tfrac{1}{24} & 0 & -\tfrac{1}{24} & 0 & 0 & 0\\
    \tfrac{1}{36} & -\tfrac{1}{24} & \tfrac{2}{27} & \tfrac{1}{36} & -\tfrac{1}{24} & 0 & 0 & 0 & 0\\
    -\tfrac{1}{27} & - \tfrac{1}{48} & -\tfrac{1}{72} & 0 & 0 & 0 & 0 & \tfrac{1}{24} & \tfrac{1}{36}\\
    -\tfrac{1}{48} & -\tfrac{1}{27} & \tfrac{1}{48} & 0 & 0 & 0 & \tfrac{1}{24} & 0 & -\tfrac{1}{24}\\
    -\tfrac{1}{72} & \tfrac{1}{48} & -\tfrac{1}{27} & 0 & 0 & 0 & \tfrac{1}{36} & -\tfrac{1}{24} & 0
\end{bmatrix}
\end{align*}
and his $L_1$ norm is $L(g_{D''})=1$.

Condition (i) in Theorem \ref{rec} is verified: 
$$(L(g_D)+\frac{1}{2})(L(g_{D'})+\frac{2}{3})-
\frac{1}{2} \cdot\frac{2}{3} = 1 \cdot \frac{4}{3}-\frac{1}{3}=1=L(g_{D''})$$

For condition (ii), we have $x=(1,1,-1)\otimes(1,1,-1)=(1,1,-1,1,1,-1-1,-1,1)$, then $f''(x)=f'(1,1,-1)f(1,1,-1)-\tfrac{1}{3}=1-\tfrac{1}{3}=\tfrac{2}{3}$, and $g_{D''}(x)=\tfrac{2}{3}$.
Therefore, $x\in (S^+\otimes T^+)$.

Similarly to how Theorem \ref{rec} defines an AND based on the Kronecker product, Theorem \ref{rec2} defines an OR based on the Kronecker product. First, we define the following terms:
\begin{itemize}
    \item $K_{R} = (1-K) I_{n}/n$,
    \item $M^{D}_{R} = K_{R} - \tfrac{1}{2}M^{D} $.
\end{itemize}

\begin{theorem}\label{rec2}
There is a WGD $D^{\prime \prime}$ such that 
\begin{itemize}
    \item[i)] $L(g_{D^{\prime \prime}}) = \vert 1-K^{\prime} \vert L(g_{D}) + \vert 1-K \vert L(g_{D^{\prime}}) + L(g_{D}) \cdot L(g_{D^{\prime}})$ and 
    \item[ii)] $g_{D^{\prime \prime}}(x) = f^{\prime \prime}(x) + K^{\prime \prime} $ for $x \in ((S^{+} \otimes T^{+}) \cup (S^{-} \otimes T^{-})\cup (S^{+} \otimes T^{-})\cup (S^{-} \otimes T^{+}))$ and some $K^{\prime \prime} \in \mathbb{R}$. Where $f^{\prime \prime} : \{1,-1\}^{nm} \rightarrow [ 0,1 ]$ satisfies
    \begin{enumerate}
        \item[(a)] $f^{\prime \prime}(\omega) = 1$ if $\omega \in ((S^{+} \otimes T^{-}) \cup (S^{-} \otimes T^{+}) \cup (S^{+} \otimes T^{+}))$, and
        \item[(b)] $f^{\prime \prime}(\omega) = 0$ if $ \omega \in (S^{-} \otimes T^{-})$
    \end{enumerate}
\end{itemize}
\end{theorem}

\begin{proof}
    Denote $M = 2(K_{R} \otimes K^{\prime}_{R} - M^{D}_{R} \otimes M^{D^{\prime}}_{R})$ as the matrix asociated to some WGD $D^{\prime \prime}$ and notice that $\tfrac{1}{2}x''M{x''}^{t} = f(x) + f^{\prime}(x') -  f(x)f^{\prime}(x') - (K + K^{\prime} - KK^{\prime})$, thus:

    \begin{itemize}
        \item if $x \in ((S^{+} \otimes T^{-}) \cup (S^{-} \otimes T^{+}) \cup (S^{+} \otimes T^{+}))$ then $g^{\prime \prime}(x) = 1 - (K + K^{\prime} - KK^{\prime})$.
        \item if $x \in (S^{-} \otimes T^{-})$ then $g^{\prime \prime}(x) = - (K + K^{\prime} - KK^{\prime})$.
    \end{itemize}
    Therefore, we've proved the preposition by taking $K^{\prime \prime} = (K + K^{\prime} - KK^{\prime})$.

    To prove the norm, we begin by noting that 
\begin{align*}
    g_{D^{\prime \prime}} &= (g_{D} + K) + (g_{D^{\prime}} + K^{\prime}) - (g_{D}+K)(g_{D^{\prime}}+K') \\
    &= (1-K^{\prime})g_{D} + (1-K)g_{D^{\prime}} - g_{D}g_{D^{\prime}}, \\
\end{align*}

Applying the norm we obtain

\begin{align*}
    L(g_{D''}) &= L((1-K^{\prime})g_{D} + (1-K)g_{D^{\prime}} - g_{D}g_{D^{\prime}}) \\
    &= L((1-K^{\prime})g_{D}) + L((1-K)g_{D^{\prime}}) + L(-g_{D}g_{D^{\prime}})) \\
    &= |1-K'|L(g_D)+|1-K|L(g_{D'})+L(g_D)\cdot L(g_{D'})
\end{align*}
\end{proof}

By taking both WDG $D$ and $D'$ depicted in Figure \ref{fig:DD'}, and by construction in proof of Theorem \ref{rec2} the following matrix representation is obtained

\begin{align*}
M^{D''} &=\begin{bmatrix}
    0 & \tfrac{1}{24} & \tfrac{1}{36} & \tfrac{1}{27} & -\tfrac{1}{24} & \tfrac{1}{36} & - \tfrac{1}{54} & \tfrac{1}{48} & \tfrac{1}{72}\\
    \tfrac{1}{24} & 0 & -\tfrac{1}{24} & \tfrac{1}{24} & \tfrac{1}{27} & -\tfrac{1}{24} & -\tfrac{1}{48} & -\tfrac{1}{54} & \tfrac{1}{48}\\
    \tfrac{1}{36} & -\tfrac{1}{24} & 0 & \tfrac{1}{36} & -\tfrac{1}{24} & \tfrac{2}{27} & -\tfrac{1}{72} & \tfrac{1}{48} & -\tfrac{1}{54}\\
    \tfrac{2}{27} & \tfrac{1}{24} & \tfrac{1}{36} & 0 & \tfrac{1}{24} & \tfrac{1}{36} & 0 & 0 & 0 \\
    -\tfrac{1}{24} & \tfrac{1}{54} & -\tfrac{1}{24} & \tfrac{1}{24} & 0 & -\tfrac{1}{24} & 0 & 0 & 0\\
    \tfrac{1}{36} & -\tfrac{1}{24} & \tfrac{2}{27} & \tfrac{1}{36} & -\tfrac{1}{24} & 0 & 0 & 0 & 0\\
    -\tfrac{1}{54} & - \tfrac{1}{48} & -\tfrac{1}{72} & 0 & 0 & 0 & 0 & \tfrac{1}{24} & \tfrac{1}{36}\\
    \tfrac{1}{48} & -\tfrac{1}{54} & \tfrac{1}{48} & 0 & 0 & 0 & \tfrac{1}{24} & 0 & -\tfrac{1}{24}\\
    \tfrac{1}{72} & \tfrac{1}{48} & -\tfrac{1}{54} & 0 & 0 & 0 & \tfrac{1}{36} & -\tfrac{1}{24} & 0
\end{bmatrix}
\end{align*}

Then
\begin{align*}
L(g_{D''})=\tfrac{5}{6}=\tfrac{1}{2}\cdot\tfrac{2}{3}+\tfrac{1}{3}\tfrac{1}{2}+\tfrac{1}{2}\cdot \tfrac{2}{3}=|1-\tfrac{1}{2}|L(g_D')+|1-\tfrac{2}{3}|L(g_{D})+L(g_D)L(g_{D'}),
\end{align*}
and condition i) holds.

Moreover, if $x=x'=(1,1,-1)$, then $f(x)=g_D(x)+\tfrac{1}{2}=1$ and 
$f'(x')=g_{D'}(x')+\tfrac{2}{3}=1$, so $x\in S^+$ and $x'\in T^+$.
Taking $x''=x\otimes x'=(1,1,-1,1,1,-1,-1,-1,1)$ and $g_{D''}(x'')=\tfrac{1}{6}$,
being
\begin{align*}
f''(x'')=f(x)+f'(x)-f(x)f'(x')=1.
\end{align*}

Since $K''= -\left(K + K'- KK'\right) = -\left(\tfrac{1}{2} + \tfrac{2}{3} - \tfrac{1}{2}\cdot\tfrac{2}{3}\right) = -\tfrac{5}{6}$, according to the Theorem \ref{rec2} we have
\begin{align*}
g_{D''}(x'')=f''(x'')+K''=1 + (-\tfrac{5}{6}) = \tfrac{1}{6}
\end{align*}
and satisfies the condition ii).

\begin{corollary}
\label{cor1}
    Let $D$ be a WDG of size $n\times n$. If $L\left ( g_{D} \right )>1$, then there is a sequence $D_{i}$ of WDGs of size $n^{i}\times n^{i}$ such that $L\left ( g_{D_{i}} \right )=\Omega \left ( L\left ( g_{D} \right )^{i} \right )$.
\end{corollary}
\begin{proof}
    Take $D_{1}=D$ and using induction obtain $D_{i+1}=D''$ by applying Theorem \ref{rec} (or Theorem \ref{rec2}) with $D'=D$.
\end{proof}
Corollary \ref{cor1} shows that the iterative application of Theorems \ref{rec} or \ref{rec2} allows us to obtain a sequence of exponentially growing WGDs, with a spectral norm $L_1$ also growing exponentially. By Theorem \ref{TH2} can be observed that if $L\left ( g_{D_{i}} \right )$ is an increasing function relative to $n$ it is possible that the sequence of quantum algorithms has an exponential advantage over the classical ones since $t=1$. 
A possible strategy for developing quantum algorithms is to first solve the optimization problems of the Definitions \ref{max} or \ref{min} using generic optimization methods \cite{absil2008optimization}, in order to generalize the algorithm through a sequence of WGDs obtained by applying Theorems \ref{rec} or \ref{rec2}.

Below we show an example of a sequence of functions obtained by a sequence of WDG constructed with the iterative application of AND operations on the Kronecker product. We also show in Corollary \ref{cor2} that the certificate complexity (see Definition \ref{def:CC}) of these functions tends to infinity in relation to the size of the input and therefore their classical randomized complexity also tends to infinity. Therefore, these iterative procedures produce functions where quantum algorithms have an exponential advantage over classical ones, since WDG represent single-query quantum algorithms.

In the following, the definition of a particular WDG class, namely $ \mathcal G^n_F $, is presented. This class is characterized by all weighted dynamical graphs $g_D:\{-1,1\}^n\to\mathbb R$
that holds $g_D(x) = 2F(x) -1$ , where 
$F$ is defined from the recursive function $F_i:\{-1,1\}^*\to \{1,0\}$ given as follows: 
\begin{equation}\label{eq:F}
F\left(\bigotimes_{j=1}^i (1,x_j) \right) = \begin{cases}1 & \text{ if }i=1\\ F\left(\bigotimes_{j=1}^{i-1} (1,x_j) \right) \cdot\left(\cfrac{x_i+1}{2}\right) & \text{ if }i> 1\end{cases}
\end{equation}

\begin{remark}
Note that $F_i$ behaves similarly to the AND function, because $F_i$ returns 1 if and only if $x_i\neq -1$ for all $i\geq 1$. So, $g_D={g_D}_{|P_S^n}$ with the set of cartesian products $$P_S^n=\bigcup_{(i_1,\dots,i_n)\in \{+,-\}^n}(S^{i_1}\times \cdots \times S^{i_n})$$ and $S^{\pm}=\{x\in\{-1,1\}^n \mid F(x)=\pm 1\}$.
\end{remark}

\begin{lemma}\label{lem:complexity}
    For each $g_D \in \mathcal G_F^n$ then $g_D(x)=2F(x)-1$ with $F(x) = 2^{-n}\prod_j^n(x_i+1)$, where $x=\bigotimes_{j=1}^n(1,x_j)$. Furthermore, their certificate complexity is $C(F(x)) = \Omega(\log(|x|))$.
\end{lemma}

\begin{proof}
    The recursion property can be proven by induction over $n\in \mathbb N$, using the fact that
    \begin{align*}
        F((1,x_1)\otimes (1,x_2))= F(1,x_1) \left(\frac{x_1 + 1}{2}\right) F(1,x_2) \left(\frac{x_2 + 1}{2}\right) = 2^{-2}(x_1+1)(x_2+1)
    \end{align*}
    by definition in Equation \ref{eq:F}.
    For the certificate complexity, defined in Section \ref{sec2}, we use the property of $F$ and doing a search for the $F$ tree to get 0 or 1. Since the function $F$ returns 1 iff $x_i=1$ for each $i\geq 1$ (see Figure \ref{fig:tree}), it suffices to perform one query for 0-certificate and the $n$-queries for 1-certificate; therefore $C(F(x))=\Omega(n)$. 
    But by construction of the variable $x$, we can obtain the operating product on the tensorial product $x=\bigotimes_{i=1}^n(1,x_i)=\left(prod(S)\right)_{S\in 2^X}$ by taking $X=\{x_1,\dots,x_n\}$ the set of variables on $x$, where $prod(S)=\prod_{y\in S} y$ if $S\neq \emptyset$ and $prod(S)=1$ otherwise. Clearly $|x|=2^n$, then $n=\log (|x|)$ -- this is the number of query for the certificate complexity. As the complexity can be queried in less time, then $C(F)\in \Omega(\log(|x|))$.
\end{proof}

\begin{figure}[ht]
\centering
\begin{tikzpicture}
  \node[circle,draw] (a) at (0,0) {$x_1$};
  \node (b) at (-1,-1) {\reflectbox{$\ddots$}};
  \node[circle,draw] (c) at (-2,-2) {$x_n$};
  \node[rectangle,draw] (d0) at (1,-1) {$0$};
  \node[rectangle,draw] (l1) at (-3,-3) {$1$};
  \node[rectangle,draw] (l0) at (-1,-3) {$0$};
  
  \draw (a) -- (b) node[above left] {$1$};
  \draw (b) -- (c) ;  
  \draw (a) -- (d0) node[pos=0.5] {$-1$};  
  \draw (c) -- (l1) ;
  \draw (c) -- (l0) node[pos=0.5] {$-1$};
\end{tikzpicture}
\caption{Scheme of the proof for certificate complexity.}
\label{fig:tree}
\end{figure}
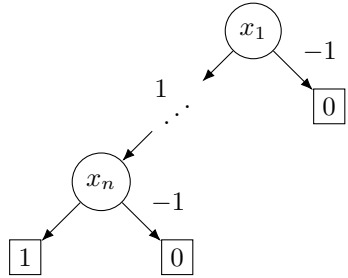

\begin{corollary}\label{cor2}  
For each $\varepsilon<1/2$, then
$R_{\varepsilon}(F) =\Omega\left(\sqrt{\log(|x|)}\right)$.
\end{corollary}

It is immediate by Equation \ref{eq:relation_R_C} and Lemma \ref{lem:complexity}.

\section{$L_1$ spectral norm and measurement}\label{sec5}

 The measurement stage of a quantum algorithm can itself be a source of computational complexity \cite{gharibian2019complexity}. Here we introduce a dimension-based measure of measurement ``capa\-city'' (Definition \ref{capac}) and link it to the $L_{1}$ norm of single-query algorithms. Theorem \ref{theo10} shows that unbounded growth of this capacity is necessary for asymptotic quantum advantage, a result that highlights measurement complexity as a fundamental resource.

In this section, we introduce a capacity measure for the measurement stage of a quantum algorithm and relate it to the $L_{1}$ norm of the output of single-query QQM algorithms. As a result, we establish a necessary condition for a single-query QQM algorithm to achieve an asymptotic advantage over classical algorithms.

The complexity of the measurement stage in quantum algorithms has been understudied in the context of quantum advantage. Here, we introduce a capacity measure based on projector dimensions, capturing how much information can be extracted in a single measurement.

\begin{definition}
\label{capac}
     Let $\mathcal{H}$ be a $n$-dimensional Hilbert space and  a CSOP  $P$ on $\mathcal{H}.$ 
     Then we denote the order of $P$ as
    \begin{equation}
    order\left( P \right)=\max_i \left\{ \min \left\{ \dim\left( P_{i} \right),n -\dim\left( P_{i} \right) \right\}\right\},
    \end{equation}
    where $\dim\left( Q \right)$ is the dimension of the subspace of $\mathcal{H}$ projected by a projector $Q$.
\end{definition}
Definition \ref{capac} formulates a capacity measure, which is applied to a CSOP, but is determined by the projector whose number of dimensions is furthest from zero or from the number of dimensions of the space where the CSOP is defined.

We denote the following:
\begin{itemize}
    \item For each $i\in\mathbb{N}$ a finite Hilbert space $\mathcal{H}_i$.
    \item For each $i\in\mathbb{N}$ a quantum query algorithm $\mathcal{A}_i$ whose states belong respectively to $\mathcal{H}_i$ and applies a single query.
    \item The final state of $\mathcal{A}_i$ on input $x$ as $\left | \Psi^{f,i}_{x}  \right> \in \mathcal{H}_i$.
    \item A finite set $Z$.
    \item For each $i$, a CSOP $P^{i}=\left\{ P^{i}_{z}:z\in Z \right\}$.
    \item For each $\mathcal{A}_i$, the probability of obtaining each output $z \in Z$ by measurement on input $x$ is given by: $\pi^{i}_{x}\left( z \right)=\left\| P^{i}_{z}\left|  \Psi^{f,i}_{x} \right\rangle \right\|^{2}.$
    \item Each $\mathcal{A}_i$ computes a function $f_i$ within error $\varepsilon$.
\end{itemize}

Then, we can state the following theorem.

\begin{theorem}\label{theo10}
If  $R_{\varepsilon}\left( f_{i} \right) = \omega\left(1\right)$ then \begin{equation}
    \lim_{i\to \infty } order\left( P^{i}_{z} \right)=\infty
\end{equation} 
 for some $z \in Z$.
\end{theorem}

\begin{proof}
    For each QQM algorithm $\mathcal{A}_i$, we denote i) $U^{i}_{j}$ as the unitary operator after the $j$-th query and ii) $P^{i}_{z}$ as the projector for the output $z$, iii) $\widetilde{U}^{i}_{n}=U^{i}_{n}U^{i}_{n-1}\ldots U^{i}_{0}$ and ${\widetilde{P}^{i}_{j,k}=\left(\widetilde{U}^{\,i}_{\,j}\right)^{\dagger}\bar{P}_{k}\widetilde{U}^{i}_{j}}$ (with $\bar{P}_{k}$ as defined in Section \ref{sec2}).
    For   each $i$ and $z\in Z$, denote $\mathcal{B}^{i}_{z}$ as a base for the subspace projected by $P^{i}_{z}$, then: \begin{equation}
        P^{i}_{z}=\sum_{b \in \mathcal{B}^{i}_{z}}\left |b\right>\left< b\right |.
    \end{equation}
    
     For each $i$, $z$, $x$, $k\in \left[n \right]$ and $b \in \mathcal{B}^{i}_{z}$, there are coefficients $\alpha^{i,z}_{k,b}$ such that 
    \begin{equation}\label{eq17}
    P^{i}_{z}U^{i}_{1}\bar{P}_{k}U^{i}_{0}\left | \Psi \right>=\sum_{b \in \mathcal{B}^{i}_{z}}\alpha^{i,z}_{k,b}\left |b\right>.
    \end{equation} 
By Theorem \ref{teo0} and Equation \ref{eq17} we have
    
\begin{equation}
    P^{i}_{z}\left | \Psi^{f,i}_{x} \right>=\sum_{k\in \left[n \right]}(-1)^{x_{k}}\left( \sum_{b \in \mathcal{B}^{i}_{z}}\alpha^{i,z}_{k,b}\left |b\right> \right),
\end{equation} which implies

\begin{equation}
    \left\| P^{i}_{z}\left | \Psi^{f,i}_{x} \right> \right\|^{2}=\sum_{k\in \left[n \right]}\sum_{h\in \left[n \right]}\left(\sum_{b \in \mathcal{B}^{i}_{z}}(-1)^{\left( x_{k}+x_{h} \right)} \alpha^{i,z}_{k,b}\alpha^{i,z}_{h,b} \right)
\end{equation}

    We define an equivalence relation  $\mathcal{R}$ in $\mathcal{S}=\left[n \right]\times \left[n \right]$, such that $\left( k,h \right)$ and $\left( k',h' \right)$ are equivalent if $(-1)^{\left( x_{k}+x_{h} \right)}=(-1)^{\left( x_{k'}+x_{h'} \right)}$ for all $x$, and denote the set of equivalence classes $\left[ \left( k,h \right) \right]$ as $\mathcal{S}/\mathcal{R}$. Therefore:

\begin{equation}
    L\left( \pi^{i}_{x}\left( z \right) \right)=\sum_{\left[ \left( k,h \right) \right]\in \mathcal{S}/\mathcal{R}}\left| \sum_{\left( k',h' \right)\in \left[ \left( k,h \right) \right]}\sum_{b \in \mathcal{B}^{i}_{z}} \alpha^{i,z}_{k',b}\alpha^{i,z}_{h',b}  \right|,
\end{equation} which implies

\begin{equation}
    L\left( \pi^{i}_{x}\left( z \right) \right) \le \sum_{b \in \mathcal{B}^{i}_{z}}\sum_{\left[ \left( k,h \right) \right]\in \mathcal{S}/\mathcal{R}}\sum_{\left( k',h' \right)\in \left[ \left( k,h \right) \right]}\left|  \alpha^{i,z}_{k',b}\alpha^{i,z}_{h',b}  \right|=\sum_{b \in \mathcal{B}^{i}_{z}}\left(  \sum_{k\in \left[n \right]}\left| \alpha^{i,z}_{k,b} \right|  \right)^{2}.
\end{equation}    
Let $\left |\alpha\right>, \left |\beta\right> \in \mathcal{H}$, and a CSOP $P=\left\{ P_{k} \right\}$ of order 1 in a Hilbert space $\mathcal{H}$; notice that if $\left\| \left | \beta\right> \right\|=1$ then 

\begin{equation}
\label{cot2}
\sum_{k} \left| \left< \beta\right |P_{k}\left |\alpha\right> \right|\le \left\| \left | \alpha\right> \right\|.  
\end{equation}
    From Equation \ref{eq17} we have
\begin{equation}
\label{eq23}
\left< b\right |U^{i}_{1}\bar{P}_{k}U^{i}_{0}\left | \Psi \right>=\alpha^{i,z}_{k,b},
\end{equation} and applying Equation \ref{cot2} on Equation \ref{eq23} we have that

\begin{equation}\label{eq24}
\sum_{k\in \left[n \right]} \left| \alpha^{i,z}_{k,b} \right|\le 1.
\end{equation}    
By Theorem \ref{TH2}, Equation \ref{eq23} and Equation \ref{eq24}; if there is some $K$ and an integer increasing infinite sequence $c_i$ satisfying 

\begin{equation}
    order\left( P^{c_i}_{z} \right)<K
\end{equation} for all $z$, then 

\begin{equation} \label{upb}
     R_{\varepsilon}\left( f_{c_i} \right)< K'
\end{equation}
for some $K'$.

Since $Z$ is finite and the probability of each $z \in Z$ can be efficiently computed by a classical algorithm according to Equation~\ref{upb}, it follows that a classical algorithm can reproduce the full output distribution within a bounded number of queries.

\end{proof}

We emphasize that Theorem \ref{theo10} is a result that applies only to single-query algorithms with a limited number of possible outputs. This result implies that the appropriate number of dimensions of the projectors for each output is a necessary condition for single-query quantum algorithms to achieve an exponential advantage over classical algorithms.

\section{Conclusion}\label{sec6}

This paper presents a computationally tractable framework for problem search where single-query quantum algorithms possess an advantage over classical decision trees. This is because single-query quantum algorithms have a straightforward relation to symmetric matrices or weighted graphs. Therefore, the problem can be attacked by combinatorial and matrix analysis techniques. The proposal is based on maximizing the $L_{1}$ spectral norm of the output of a multilinear second degree polynomial as a weighted graph, since a high $L_{1}$ spectral norm is a necessary condition for such an output to be costly to simulate with classical decision trees. Accordingly, the proposal can be formulated as a maximization problem with constraints, as well as an equivalent minimization problem. However, non-optimal solutions that tend to increase the spectral norm $L_{1}$ are also interesting. As future work, we propose developing a simulation technique to transform a WDG into a QQM algorithm, such that if the WDG error tends to zero when computing a function, then the QQM algorithm error also tends to zero.

The strategies presented to maximize the $L_1$ norm of second-degree multilinear polynomials depend on expressions with Kronecker products between the matrix representations of two polynomials. As we saw when analyzing the maxima and minima of polynomials with maximized $L_1$ norm, we identified Boolean functions where single-query quantum algorithms can outperform classical ones. An open question is what strategies we can formulate to maximize the $L_1$ norm that involve other types of composition or that do not even depend on any composition?

Finally, we present a necessary condition for a single-query quantum algorithm with a finite number of possible outputs to compute functions with an asymptotic advantage over classical query algorithms. This condition is based on a measure that depends on the number of dimensions of the output projections, which presents the complexity of the measurement stage as a computational resource for QQM algorithms and represents a guideline in the design of quantum algorithms. An obvious question regarding this result is the impact of measurement complexity for a sequence of algorithms whose number of queries is bounded by a constant greater than 1 and, more generally, by a function that depends on the size of the input.

 While the weighted graph formulation offers a novel approach for analyzing single-query quantum advantage, several directions warrant further investigation. First, the approach assumes bounded-error approximation; its extension to zero-error or exact computation remains open. Second, the framework currently addresses only degree-$2$ polynomials; generalizing to higher-degree polynomials may require hypergraph representations. Finally, the connection between WDG optimization and implementable quantum circuits is not yet constructive, a crucial gap for practical applications.

\bmhead{Acknowledgements}

This work was supported by CONACYT, Paraguay, under Grant PINV01-397 and by the Iberoamerican Program for Science and Technology for Development (CYTED) through RIPAISC - Iberoamerican Network for the Advancement of Quantum Software Engineering (525RT0174). F.L.M. acknowledges financial support of CNPq process No. 407296/2021-2, CNPq process No. 306049/2025-2, and of the National Institute of Science and Technology for Applied Quantum Computing through CNPq process No. 408884/2024-0.

\bibliography{sn-bibliography}

\end{document}